\documentclass{llncs}

\newcommand{\tit}{Phylogenetic Networks Do not Need to Be Complex:\\ Using Fewer Reticulations to Represent Conflicting Clusters}

\newcommand{\cass}{\textsc{Cass}}
\newcommand{\cset}{\ensuremath{\mathcal{C}}}

\usepackage{amsmath}
\usepackage{tabularx}
\newcolumntype{R}{>{\raggedleft\arraybackslash}X}
\usepackage{url}
\usepackage{amssymb}
\usepackage{amsfonts}
\usepackage[british]{babel}
\usepackage{graphicx}
\usepackage[pdftex, colorlinks=true, linkcolor=blue, citecolor=blue, urlcolor=blue, hypertexnames=true, letterpaper=true, pdfstartview=FitH]{hyperref}
\hypersetup{ %
pdfauthor = {Leo van Iersel, Steven Kelk, Regula Rupp and Daniel Huson}, %
pdftitle = {\tit}, %
pdfsubject  = {Phylogenetic Network Construction}, %
pdfkeywords = {clusters, phylogenetic networks, polynomial-time algorithm}, %
}
\usepackage[letterpaper,hmargin=1in,vmargin=1in]{geometry}
\usepackage{subfigure}
\usepackage{algorithm}
\usepackage[noend]{algorithmic}

\begin{document}

\title{\tit}
\author{Leo van Iersel\inst{1}, Steven Kelk\inst{2}, Regula Rupp\inst{3} and Daniel Huson\inst{3}}
\institute{University of Canterbury, Department of Mathematics and Statistics,\\Private Bag 4800, Christchurch, New Zealand. l.j.j.v.iersel@gmail.com \\ \and Centrum voor Wiskunde en Informatica (CWI)\\ P.O. Box 94079,
1090 GB Amsterdam, The Netherlands. s.m.kelk@cwi.nl \\ \and Center for Bioinformatics ZBIT, T\"ubingen University\\ Sand 14, 72076 T\"ubingen, Germany. \{huson,rrupp\}@informatik.uni-tuebingen.de}
\maketitle

\begin{abstract}
Phylogenetic trees are widely used to display estimates of how groups of species evolved. Each phylogenetic tree can be seen as a collection of \emph{clusters}, subgroups of the species that evolved from a common ancestor. When phylogenetic trees are obtained for several data sets (e.g. for different genes), then their clusters are often contradicting. Consequently, the set of all clusters of such a data set cannot be combined into a single phylogenetic tree. \emph{Phylogenetic networks} are a generalization of phylogenetic trees that can be used to display more complex evolutionary histories, including \emph{reticulate events} such as hybridizations, recombinations and horizontal gene transfers. Here we present the new \textsc{Cass} algorithm that can combine any set of clusters into a phylogenetic network. We show that the networks constructed by \textsc{Cass} are usually simpler than networks constructed by other available methods. Moreover, we show that \textsc{Cass} is guaranteed to produce a network with at most two reticulations per biconnected component, whenever such a network exists. We have implemented \textsc{Cass} and integrated it in the freely available Dendroscope software.
\end{abstract}

\section{Introduction}
\emph{Phylogenetics} studies the reconstruction of evolutionary histories from genetic data of currently living organisms. A (rooted) \emph{phylogenetic tree} is a representation of such an evolutionary history in which species evolve by mutation and speciation. The leaves of the tree represent the species under consideration and the root of the tree represents their most recent common ancestor. Each internal node represents a speciation: one species splits into several new species. Thus, mathematically speaking, such a node has indegree one and outdegree at least two. In recent years, a lot of work has been done on developing methods for computing (rooted) \emph{phylogenetic networks}~\cite{Gambette2009,HusonRupp2009}, which form a generalization of phylogenetic trees. Next to nodes representing speciation, rooted phylogenetic networks can also contain \emph{reticulations}: nodes with indegree at least two. Such nodes can be used to represent recombinations, hybridizations or horizontal gene transfers, depending on the biological context. In addition, phylogenetic networks can also be interpreted in a more abstract sense, as a visualization of contradictory phylogenetic information in a single diagram.

Suppose we wish to investigate the evolution of a set~$\mathcal{X}$ of taxa (e.g. species or strains). Each edge of a rooted phylogenetic tree represents a \emph{cluster}: a proper subset of the taxon set~$\mathcal{X}$. In more detail, an edge~$(u,v)$ represents the cluster containing those taxa that are descendants of~$v$. Each phylogenetic tree~$T$ is uniquely defined by the set of clusters represented by~$T$. Phylogenetic networks also represent clusters. Each of their edges represents one ``hardwired'' and at least one ``softwired'' cluster. An edge~$(u,v)$ of a phylogenetic network \emph{represents} a cluster~$C\subset\mathcal{X}$ \emph{in the hardwired sense} if~$C$ equals the set of taxa that are descendants of~$v$. Furthermore,~$(u,v)$ \emph{represents}~$C$ \emph{in the softwired sense} if~$C$ equals the set of all taxa that can be reached from~$v$ when, for each reticulation~$r$, exactly one incoming edge of~$r$ is ``switched on'' and the other incoming edges of~$r$ are ``switched off''. An equivalent definition states that a phylogenetic network~$N$ \emph{represents} a cluster~$C$ \emph{in the softwired sense} if there exists a tree~$T$ that is displayed by~$N$ (formally defined below) and represents~$C$. In this paper we will always use ``represent'' in the softwired sense. It is usually the clusters in a tree that are of more interest, and less the actual trees themselves, as clusters represent putative monophyletic groups of related species. For a complete introduction to clusters see Huson and Rupp~\cite{HusonRupp2009}.

In phylogenetic analysis, it is common to compute phylogenetic trees for more than one data set. For example, a phylogenetic tree can be constructed for each gene separately, or several phylogenetic trees can be constructed using different methods. To accurately reconstruct the evolutionary history of all considered taxa, one would preferably like to use the set~$\mathcal{C}$ of all clusters represented by at least one of the constructed phylogenetic trees. In general however, some of the clusters of the different trees will be conflicting, which means that there will be no single phylogenetic tree representing~$\mathcal{C}$. Therefore, several recent publications have studied the construction of a phylogenetic \emph{network} representing~$\mathcal{C}$. Huson and Rupp~\cite{HusonRupp2008} describe how a phylogenetic network can be constructed that represents~$\mathcal{C}$ in the hardwired sense (a \emph{cluster network}). A network is a \emph{galled network} if it contains no path between two reticulations that is contained in a single \emph{biconnected component} (a maximal subgraph that cannot be disconnected by removing a single node). Huson and Kl\"opper~\cite{HusonKloepper2007} and Huson et al.~\cite{HusonEtAl2009} describe an algorithm for constructing a galled network representing~$\mathcal{C}$ in the softwired sense.

Related literature describes the construction of phylogenetic networks from phylogenetic trees or \emph{triplets} (phylogenetic trees on three taxa). A tree or triplet~$T$ is \emph{displayed} by a network~$N$ if there is a subgraph~$T'$ of~$N$ that is a subdivision of~$T$ (i.e.~$T'$ can be obtained from~$T$ by replacing edges by directed paths). Computing the minimum number of reticulations required in a phylogenetic network displaying two input trees (on the same set of taxa) was shown to be APX-hard by Bordewich and Semple~\cite{BordewichSemple2007}. Bordewich et al.~\cite{BordewichEtAl2007} proposed an exact exponential-time algorithm for this problem and Linz and Semple~\cite{LinzSemple2009} showed that it is fixed parameter tractable (FPT), if parameterized by the minimum number of reticulations. The downside of these algorithms is that they are very rigid in the sense that one generally needs very complex networks in order to display the given trees.

The \emph{level} of a binary network is the maximum number of reticulations in a biconnected component\footnote{In Section~\ref{sec:prelim} we generalize the notion of \emph{level} to non-binary networks.}, and thus provides a measure of network complexity. Given an arbitrary number of trees on the same set of taxa, Huynh et al.~\cite{HuynhEtAl2005} describe a polynomial-time algorithm that constructs a level-1 phylogenetic network that displays all trees and has a minimum number of reticulations, if such a network exists (which is unlikely in practice). Given a triplet for each combination of three taxa, Jansson, Sung and Nguyen~\cite{JanssonEtAl2006,JanssonSung2006} give a polynomial-time algorithm that constructs a level-1 network displaying all triplets, if such a network exists. The algorithm by van Iersel and Kelk~\cite{simplicity} can be used to find such a network that also minimizes the number of reticulations. These results have later been extended to level-2~\cite{TCBB2009,simplicity} and more recently to level-$k$, for all~$k\in\mathbb{N}$~\cite{ToHabib2009}. Although this work on triplets is theoretically interesting, it has the practical drawback that biologists are not interested in triplets (but rather in trees or clusters) and that these algorithms need a triplet for each combination of three taxa as input, while some triplets might be difficult to derive correctly.

\begin{figure}[t]
  \centering
  \begin{subfigure}[]
  {
    \centering
    \includegraphics[scale=0.5]{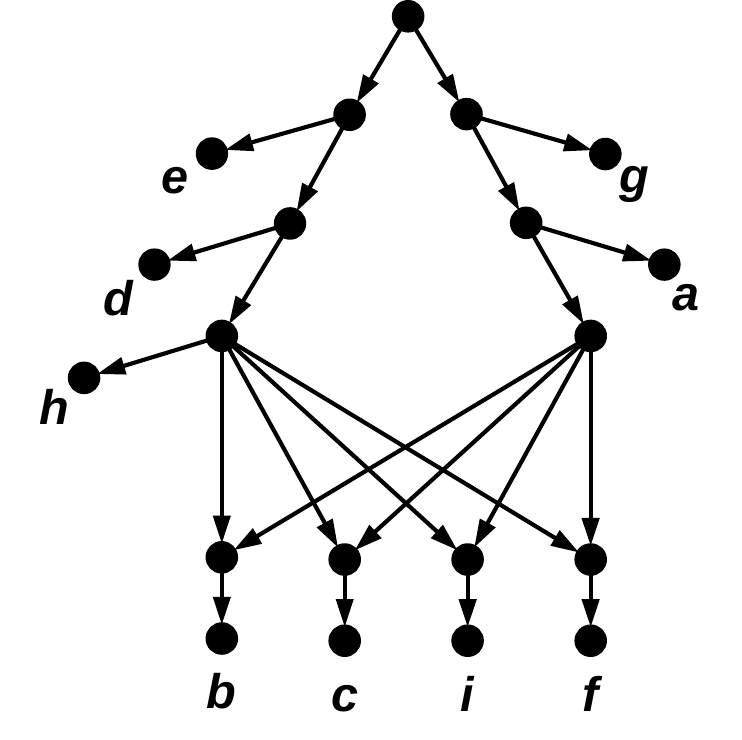}
    \label{fig:gallednetwork}
    \vspace{.5cm}
  }
  \end{subfigure}
  \hspace{3cm}
  \begin{subfigure}[]
  {
    \centering
    \includegraphics[scale=0.5]{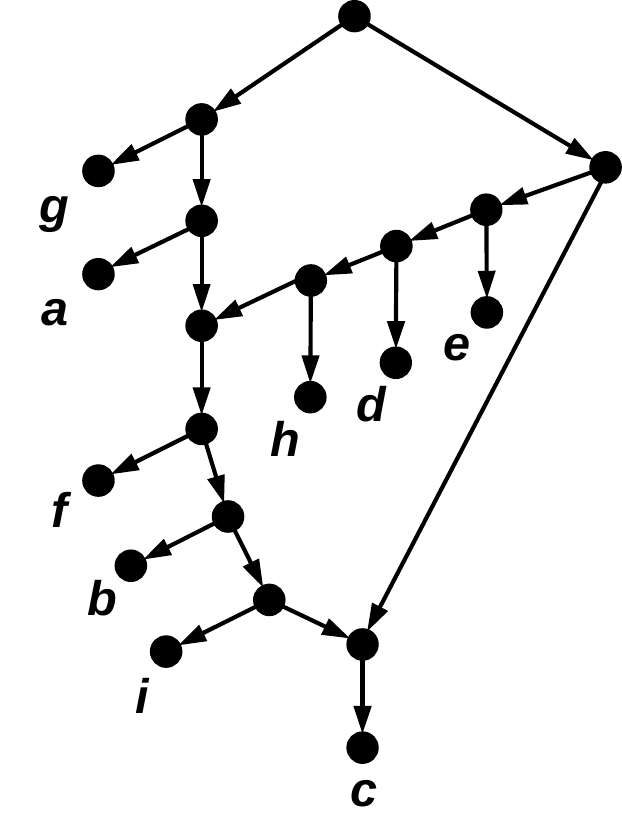}
    \hspace{1cm}
    \label{fig:level2example}
  }
  \end{subfigure}
  \caption{(a) The output of the galled network algorithm~\cite{HusonEtAl2009} for $\mathcal{C}=\{\{a,b,f,g,i\}$, $\{a,b,c,f,g,i\}$, $\{a,b,f,i\}$, $\{b,c,f,i\}$, $\{c,d,e,h\}$, $\{d,e,h\}$, $\{b,c,f,h,i\}$, $\{b,c,d,f,h,i\}$, $\{b,c,i\}$, $\{a,g\}$, $\{b,i\}$, $\{c,i\}$, $\{d,h\}\}$ and (b) the network constructed by \textsc{Cass} for the same input.}
  \label{fig:example}
\end{figure}
In this article, we present the algorithm \textsc{Cass}\footnote{Named after the Cass Field Station in New Zealand.}, which takes any set~$\mathcal{C}$ of clusters as input and constructs a phylogenetic network that represents~$\mathcal{C}$. Furthermore, the algorithm aims at minimizing the level of the constructed network and in this sense \textsc{Cass} is the first algorithm to combine the flexibility of clusters with the power of level minimization. \textsc{Cass} constructs a phylogenetic tree representing~$\mathcal{C}$ whenever such a tree exists. Moreover, we prove that \textsc{Cass} constructs a level-1 or level-2 network representing~$\mathcal{C}$ whenever there exists a level-1 or level-2 network representing~$\mathcal{C}$, respectively. Experimental results show that also when no level-2 network representing~$\mathcal{C}$ exists, \textsc{Cass} usually constructs a network with a significantly lower level and lower number of reticulations compared to other algorithms. In fact, we conjecture that similar arguments as in our proof for level-2 can be used to show that \textsc{Cass} always constructs a level-$k$ network with minimum~$k$. We prove a decomposition theorem for level-$k$ networks that supports this conjecture. Finally, we prove that \textsc{Cass} runs in polynomial time if the level of the output network is bounded by a constant.

We have implemented \textsc{Cass} and added it to our popular tree-drawing program Dendroscope~\cite{HusonEtAl2007}, where it can be used as an alternative for the cluster network~\cite{HusonRupp2008} and galled network~\cite{HusonEtAl2009} algorithms. Experiments show that, although \textsc{Cass} needs more time than these other algorithms, it constructs a simpler network representing the same set of clusters. For example, Figure~\ref{fig:gallednetwork} shows a set of clusters and the galled network with four reticulations constructed by the algorithm in~\cite{HusonEtAl2009}. However, for this data set also a level-2 network with two reticulations exists, and \textsc{Cass} can be used to find this network, see Figure~\ref{fig:level2example}. Dendroscope now combines the powers of \textsc{Cass} and the two previously existing algorithms for constructing galled- and cluster networks.

\section{${\mbox{Level-}k}$ Networks and Clusters}\label{sec:prelim}
Consider a set~$\mathcal{X}$ of taxa. A \emph{(phylogenetic) network} (on~$\mathcal{X}$) is a directed acyclic graph with a single root and leaves bijectively labeled by~$\mathcal{X}$. The indegree of a node~$v$ is denoted~$\delta^-(v)$ and~$v$ is called a \emph{reticulation} if~$\delta^-(v)\geq 2$. An edge~$(u,v)$ is called a \emph{reticulation edge} if its head~$v$ is a reticulation and is called a \emph{tree edge} otherwise. The \emph{reticulation number} of a phylogenetic network~$N=(V,E)$ is defined as
\[\sum_{\substack{v\in V: \delta^-(v)>0}}(\delta^-(v)-1) = |E| - |V| + 1 \enspace.\]

A directed acyclic graph is \emph{connected} (also called ``weakly connected'') if there is an
undirected path (ignoring edge orientations) between each pair of nodes. A node (edge) of a directed graph is called a \emph{cut-node} (\emph{cut-edge}) if its removal disconnects the graph. A directed graph is \emph{biconnected} if it contains no cut-nodes. A biconnected subgraph~$B$ of a directed graph~$G$ is said to be a \emph{biconnected component} if there is no biconnected subgraph~$B' \neq B$ of $G$ that contains~$B$.

A phylogenetic network is said to be a \emph{${\mbox{level-}k}$ network} if each biconnected component has reticulation number at most~$k$.\footnote{Note that to determine the reticulation number of a biconnected component one only counts edges inside this biconnected component.} A phylogenetic network is called \emph{binary} if each node has either indegree at most one and outdegree at most two or indegree at most two and outdegree at most one. Note that the above definition of level generalizes the original definition~\cite{ChoyEtAl2005} for binary networks. A ${\mbox{level-}k}$ network is called a \emph{simple level-$\leq k$ network} if the head of each cut-edge is a leaf. A simple level-$\leq k$ network is called a \emph{simple level-$k$ network} if its reticulation number is precisely~$k$. A \emph{phylogenetic tree} (on~$\mathcal{X}$) is a phylogenetic network (on~$\mathcal{X}$) without reticulations, i.e. a level-0 network.

Consider a set of taxa~$\mathcal{X}$. Proper subsets of~$\mathcal{X}$ are called \emph{clusters}. We say that two clusters~$C_1,C_2\subset\mathcal{X}$ are \emph{compatible} if either~$C_1\cap C_2=\emptyset$ or~$C_1\subset C_2$ or~$C_2\subset C_1$. Consider a set of clusters~$\mathcal{C}$. We say that a set of taxa~$X\subset\mathcal{X}$ is \emph{separated} (by~$\mathcal{C}$) if there exists a cluster~$C\in\mathcal{C}$ that is incompatible with~$X$. The \emph{incompatibility graph}~$IG(\mathcal{C})$ of~$\mathcal{C}$ is the undirected graph~$(V,E)$ that has node set~$V=\mathcal{C}$ and edge set
\[{E=\{\{C_1,C_2\}\enspace |\enspace C_1 \mbox{ and } C_2\mbox{ are incompatible clusters in } \mathcal{C}\}}\enspace . \]

\section{Decomposing ${\mbox{Level-}k}$ Networks}\label{sec:decomp}
In this section, we describe the general outline of our algorithm \textsc{Cass}. We show how the problem of determining a ${\mbox{level-}k}$ network can be decomposed into a set of smaller problems by examining the incompatibility graph. Our algorithm will first construct a simple level-$\leq k$ network for each connected component of the incompatibility graph and subsequently merge these simple level-$\leq k$ networks into a single ${\mbox{level-}k}$ network on all taxa.

Consider a set of taxa~$\mathcal{X}$ and a set~$\mathcal{C}$ of input clusters. We assume that all singletons (sets $\{x\}$ with~$x\in \mathcal{X}$) are clusters in~$\mathcal{C}$. Our algorithm proceeds as follows.

\noindent\textbf{Step 1.} Find the nontrivial connected components~$\mathcal{C}_1,\ldots,\mathcal{C}_p$ of the incompatibility graph~$IG(\mathcal{C})$. For each~$i\in\{1,\ldots,p\}$, let~${\mathcal{C}_{i}}'$ be the result of collapsing unseparated sets of taxa as follows. Let $\mathcal{X}_i=\bigcup_{C\in \mathcal{C}_i} C$. For each maximal subset~$X\subset\mathcal{X}_i$ that is not separated by~$\mathcal{C}_i$, replace, in each cluster in~$\mathcal{C}_i$, the elements of~$X$ by a single new taxon~$X$, e.g. if~$X=\{b,c\}$ then a cluster~$\{a,b,c,d\}$ is modified to~$\{a, \{b,c\}, d\}$.

\noindent\textbf{Step 2.} For each~${i\in\{1,\ldots,p\}}$, construct a simple level-$\leq k$ network~$N_i$ representing~${\mathcal{C}_{i}}'$.

\noindent\textbf{Step 3.} Let~$\mathcal{C}^*$ be the result of applying the following modifications to~$\mathcal{C}$, for each~$i\in\{1,\ldots,p\}$: remove all clusters that are in~$\mathcal{C}_i$, add a cluster~$\mathcal{X}_i$ and add each maximal subset~$X\subset\mathcal{X}_i$ that is not separated by~$\mathcal{C}_i$. Construct the unique phylogenetic tree~$T$ on~$\mathcal{X}$ representing precisely those clusters in~$\mathcal{C}^*$.

\noindent\textbf{Step 4.} For each~$i\in\{1,\ldots,p\}$, replace in~$T$ the lowest common ancestor~$v_i$ of~$\mathcal{X}_i$ by the simple level-$\leq k$ network~$N_i$ as follows. Delete all edges leaving~$v_i$ and merge~$T$ with~$N_i$ by identifying the root of~$N_i$ with~$v_i$ and identifying each leaf of~$N_i$ labeled~$X$ by the lowest common ancestor of the leaves labeled~$X$ in~$T$. Output the resulting network.

Notice that Steps~1,3 and 4 are similar to the corresponding steps in algorithms for constructing galled trees (i.e. level-1 networks) and galled networks~\cite{HusonKloepper2007,HusonRupp2009,HusonEtAl2009}. The reason why we use the same set-up in our algorithm, is outlined by Theorem~\ref{thm:decomp}. It shows that, when constructing a ${\mbox{level-}k}$ network displaying a set of clusters, we can restrict our attention to ${\mbox{level-}k}$ networks that satisfy the \emph{decomposition property}~\cite{HusonRupp2009}, the definition of which we repeat below.

Because a cluster~$C\in\mathcal{C}$ can be represented by more than one edge in a network~$N$, an \emph{edge assignment}~$\epsilon$ is defined as a mapping that chooses for each cluster~$C\in\mathcal{C}$ a single tree edge~$\epsilon(C)$ of~$N$ that represents~$C$. A network~$N$ representing~$\mathcal{C}$ is said to satisfy the \emph{decomposition property} w.r.t.~$\mathcal{C}$ if there exists an edge assignment~$\epsilon$ such that:
\begin{itemize}
\item[$\bullet$] for any two clusters~${C_1,C_2\in\mathcal{C}}$, the edges~$\epsilon(C_1)$ and~$\epsilon(C_2)$ are contained in the same biconnected component of~$N$ if and only if~$C_1$ and~$C_2$ lie in the same connected component of the incompatibility graph~$IG(\mathcal{C})$.
\end{itemize}
\begin{theorem}\label{thm:decomp} Let~$\mathcal{C}$ be a set of clusters. If there exists a ${\mbox{level-}k}$ network representing~$\mathcal{C}$, then there also exists such a network satisfying the decomposition property w.r.t.~$\mathcal{C}$.
\end{theorem}
\begin{proof}
Let~$\mathcal{C}$ be a set of input clusters and~$N$ a ${\mbox{level-}k}$ network representing~$\mathcal{C}$. Let~$\mathcal{C}_1,\ldots,\mathcal{C}_p$ be the nontrivial connected components of the incompatibility graph~$IG(\mathcal{C})$. For each~$i\in\{1,\ldots,p\}$, we construct a simple level-$\leq k$ network~$N_i$ as follows. Let $\mathcal{X}_i=\bigcup_{C\in \mathcal{C}_i} C$ as before. For each maximal subset~$X\subset\mathcal{X}_i$ (with~$|X|>1$) that is not separated by~$\mathcal{C}_i$, replace in~$N$ an arbitrary leaf labeled by an element of~$X$ by a leaf labeled~$X$ and remove all other leaves labeled by elements of~$X$. In addition, remove all leaves with labels that are not in~$X_i$. We tidy up the resulting graph by repeatedly applying the following five steps until none is applicable: (1) delete unlabeled nodes with outdegree 0; (2) suppress nodes with indegree and outdegree 1 (i.e. contract one edge incident to the node); (3) replace multiple edges by single edges, (4) remove the root if it has outdegree 1 and (5) contract biconnected components that have only one outgoing edge. This leads to a ${\mbox{level-}k}$ network~$N_i$. Let~${\mathcal{C}_{i}}'$ be defined as in Step~1 of the algorithm. By its construction,~$N_i$ represents~${\mathcal{C}_i}'$. Furthermore,~$N_i$ is a simple level-$\leq k$ network, because if it would contain a cut-edge~$e$ whose head is not a leaf, then the set of taxa labeling leaves reachable from~$e$ would not be separated by~${\mathcal{C}_{i}}'$ and would hence have been collapsed. Finally, the networks~$N_1,\ldots,N_p$ can be merged into a ${\mbox{level-}k}$ network representing~$\mathcal{C}$ and satisfying the decomposition property by executing Steps~3~and~4 of the algorithm. \qed
\end{proof}
Note that the analogous statement obtained by replacing ``${\mbox{level-}k}$ network'' by `` network with~$k$ reticulations'' does not hold, as shown in~\cite{HusonEtAl2009}, based on~\cite{GusfieldEtAl2007}.
\section{Simple Level-$k$ Networks}\label{sec:SLk}
This section describes how one can construct a simple level-$k$ network representing a given set of clusters. We say that a phylogenetic tree~$T$ is a \emph{strict subtree} of a network~$N$ if~$T$ is a subgraph of~$N$ and for each node~$v$ of~$T$, except its root, it holds that the in- and outdegree of~$v$ in~$T$ are equal to the in- and (respectively) outdegree of~$v$ in~$N$.

Informally, our method for constructing simple level-$k$ networks operates as follows. We loop through all taxa~$x$. For each choice for~$x$, we remove it from each cluster and subsequently collapse all maximal ``ST-sets'' (``strict tree sets'', defined below) of the resulting cluster set. We repeat this step~$k$ times. The idea behind this strategy is as follows. Observe that any simple level-$k$ network~$N$ contains a leaf whose parent is a reticulation. If we would remove this leaf and reticulation from~$N$, the resulting network might contain one or more strict subtrees. Each such strict subtree corresponds to an ST-set. Moreover, for the case~$k\leq 2$ we prove that (without loss of generality) each maximal strict subtree corresponds to a maximal ST-set. Collapsing each maximal strict subtree of the network would lead to a (not necessarily simple) ${\mbox{level-}(k-1)}$ network, which would again contain a leaf whose parent is a reticulation. It follows that we can indeed repeat the described steps~$k$ times, after which all leaves will be collapsed into just two taxa and the second phase of the algorithm starts.

We create a network consisting of a root with two children, labeled by the only two taxa. Then we ``decollapse'', i.e. we replace each leaf labeled by an ST-set by a strict subtree. Subsequently we add a new leaf below a new reticulation and label it by the latest removed taxon. Since we do not know where to create the new reticulation, we try adding the reticulation below each pair of edges. For each constructed simple level-$k$ network, we check whether it represents all input clusters. If it does, we output the resulting network, after contracting any edges that connect two reticulations.

\begin{figure}[t]
  \centering
  \includegraphics[width=.75\textwidth]{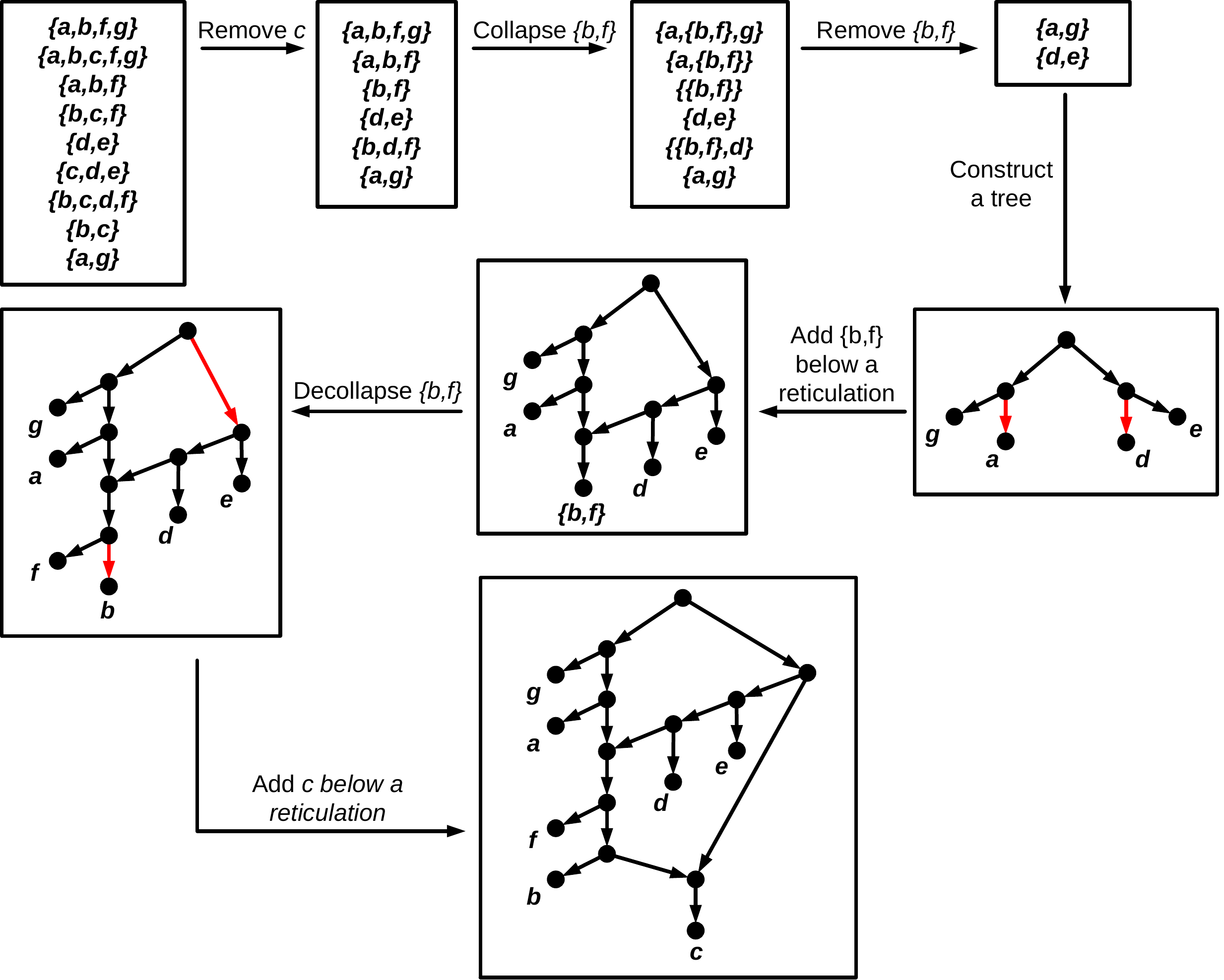}
  \caption{Construction of a simple level-2 network by the \textsc{Cass} algorithm. The edges~$e_1,e_2$ that will be subdivided are colored red. Singleton clusters have been omitted, as well as the last collapse-step, for simplicity.}
  \label{fig:fullexample}
\end{figure}
Let us now formalize this algorithm. Given a set~$S\subseteq\mathcal{X}$ of taxa, we use~$\mathcal{C}\setminus S$ to denote the result of removing all elements of~$S$ from each cluster in~$\mathcal{C}$ and we use~$\mathcal{C}|S$ to denote~${\mathcal{C}\setminus (\mathcal{X}\setminus S)}$ (the restriction of~$\mathcal{C}$ to~$S$). We say that a set~$S\neq\mathcal{X}$ is an \emph{ST-set} (strict tree set) w.r.t.~$\mathcal{C}$, if~$S$ is not separated by~$\mathcal{C}$ and any two clusters~$C_1,C_2\in\mathcal{C}|S$ are compatible. An ST-set~$S$ is \emph{maximal} if there is no ST-set~$T$ with~$S\subsetneq T$. Informally, the maximal ST-sets are the result of repeatedly collapsing pairs of unseparated taxa as long as possible.

We use~\textsc{Collapse}$(\mathcal{C})$ to denote the result of collapsing each maximal ST-set~$S$ into a single taxon~$S$. More precisely, for each cluster~$C\in\mathcal{C}$ and maximal ST-set~$S$ of~$\mathcal{C}$, we replace~$C$ by~${C\setminus S\cup\{\{S\}\}}$. For example (omitting singleton clusters), if
\[ \mathcal{C}=\{\enspace\{1,2\},\enspace\{2,3,4\},\enspace\{3,4\}\enspace\}\enspace ,\]
then~$\{3,4\}$ is the only nonsingleton maximal ST-set and
\[\mbox{\textsc{Collapse}}(\mathcal{C})=\{\enspace\{1,2\},\enspace\{2,\{3,4\}\}\enspace\}\enspace .\]
The set of taxa of a (collapsed) cluster set~$\mathcal{C}$ is denoted~$\mathcal{X}(\mathcal{C})$. Thus, for the above example, $\mathcal{X}(\mbox{\textsc{Collapse}}(\mathcal{C}))=\{1,2,\{3,4\}\}$. We are now ready to give the pseudocode of {\cass}$(k)$ in Algorithm~\ref{alg:SLk}. The actual implementation is slightly more complex and much more space efficient.

\begin{algorithm}[t]\caption{{\cass}$(k)$: constructing a simple level-$k$ network from clusters}\label{alg:SLk}
\begin{algorithmic} [1]
\STATE\textbf{input} $(\mathcal{C},\mathcal{X},k,k')$
\STATE\textbf{output} \textsc{Cass}$(\mathcal{C},\mathcal{X},k,k')$
\STATE // in the initial call to the algorithm,~$k'=k$
\STATE $\mathcal{N}:=\emptyset$
\IF{$k'=0$}
\STATE \textbf{return} the unique tree representing exactly those clusters in~$\mathcal{C}$ or return~$\emptyset$ if no such tree exists
\ENDIF
\FOR{$x\in\mathcal{X}\cup\{d\}$}
\STATE // $d$ is a dummy taxon not in~$\mathcal{X}$
\STATE \textbf{remove leaf:} $\mathcal{C}':=\mathcal{C}\setminus\{x\}$
\STATE \textbf{collapse:} $\mathcal{C}'':=\mbox{\sc{Collapse}}(\mathcal{C}')$
\STATE \textbf{recurse:} $\mathcal{N}':=$ \textsc{Cass}$(\mathcal{C}'',\mathcal{X}(\mathcal{C}''),k,k'-1)$
\FOR{$N'\in\mathcal{N}'$}
\STATE \textbf{decollapse:} replace each leaf of~$N'$ labeled by a maximal ST-set~$S$ w.r.t.~$\mathcal{C}'$ by the tree on~$S$ representing exactly those clusters in~$\mathcal{C}'|S$
\FOR{each pair of edges~$e_1,e_2$}
\STATE \textbf{add leaf below reticulation:} create a reticulation~$t$, a leaf~$l$ labeled~$x$ and an edge from~$t$ to~$l$;
\STATE for~$i=1,2$, insert a node~$v_i$ into~$e_i$ and add an edge from~$v_i$ to~$t$, this gives network~$N$
\IF{$N$ represents~$\mathcal{C}$}
\STATE \textbf{save network:} $\mathcal{N}:=\mathcal{N}\cup\{N\}$
\ENDIF
\ENDFOR
\ENDFOR
\ENDFOR
\IF{$k=k'$}
\STATE\textbf{return} any simple level-$k$ network in $\mathcal{N}$, after removing each leaf labeled~$d$ and contracting each edge connecting two reticulations
\ELSE
\STATE\textbf{return} $\mathcal{N}$
\ENDIF
\end{algorithmic}
\end{algorithm}
Figure~\ref{fig:fullexample} shows how the {\cass}$(2)$ algorithm for example constructs a simple level-2 network. We will now show that {\cass}$(1)$ and {\cass}$(2)$ will indeed construct a simple level-1 respectively level-2 network whenever this is possible.

\begin{lemma}\label{lem:level2} Given a set of clusters~$\mathcal{C}$, such that~$IG(\mathcal{C})$ is connected and any~$X\subsetneq\mathcal{X}$ is separated, {\cass}$(1)$ and {\cass}$(2)$ construct a simple level-1 respectively a simple level-2 network representing~$\mathcal{C}$, if such a network exists.
\end{lemma}
\begin{proof}
The general idea of the proof is as follows. Details have been omitted due to space constraints. Assume~$k\leq 2$. It is clear that any (simple) level-$k$ network~$N$ contains a reticulation~$r$ with a leaf, say labeled~$x$, as child. Let $N\setminus \{x\}$ denote the network obtained by removing the reticulation~$r$ and the leaf labeled~$x$ from~$N$. This network might contain one or more strict subtrees. By the definition of ST-set, the set of leaf-labels of each maximal strict subtree corresponds to an ST-set w.r.t. $\mathcal{C}\setminus\{x\}$. However, in general not each such set needs to be a \emph{maximal} ST-set. This is critical, because the total number of ST-sets can be exponentially large. Therefore, the main ingredient of our proof is the following. We show that whenever there exists a simple level-$k$ network representing~$\mathcal{C}$, there exists a simple level-$k$ network~$N'$ representing~$\mathcal{C}$ such that the sets of leaf-labels of the maximal strict subtrees of~$N'\setminus\{x\}$ are the maximal ST-sets w.r.t. $\mathcal{C}\setminus\{x\}$, with~$x$ the label of some leaf whose parent is a reticulation in~$N'$. This is clearly true for~$k=1$. For~$k=2$ we sketch our proof below.

Let us first mention that the actual algorithm is slightly more complicated than the pseudocode in Algorithm~\ref{alg:SLk}. Firstly, when {\cass}($k$) constructs a tree, it adds a new ``dummy'' root to this tree and creates an edge from this dummy root to the old root. Such a dummy root is removed before outputting a network. Secondly, whenever the algorithm removes a dummy taxon~$d$, it makes sure that it does not collapse in the previous step.

Suppose there exists some level-2 network representing~$\mathcal{C}$. It can be shown that any such network is simple and that there exists at least one binary such network, say~$N$. Since~$N$ is a binary simple ${\mbox{level-}2}$ network, there are only four possibilities for the structure of~$N$ (after removing leaves), see~\cite{TCBB2009}. These structures are called \emph{generators}. In each case, $N\setminus\{x\}$ contains at most two maximal strict subtrees that have more than one leaf. Furthermore, $N\setminus\{x\}$ contains exactly one reticulation~$r'$, below which hangs a strict subtree~$T_r$ with set of leaf-labels~$X_r$ (possibly, $|X_r|=1$ or~$|X_r|=0$).

First we assume that~$X_r$ is not a maximal ST-set w.r.t. $\mathcal{C}\setminus\{x\}$. In that case it follows that there is some maximal ST-set~$X$ that contains~$X_r$ and also contains at least one taxon labeling a leaf~$\ell$ that is not reachable by a directed path from the reticulation of $N\setminus\{x\}$. We can replace~$\ell$ by a strict subtree on~$X$ that represents $\mathcal{C}|X$. Such a tree exists because~$X$ is an ST-set. We remove all leaves that label elements of~$X$ and are not in this strict subtree. Since there are now no leaves left below the reticulation, we can remove this reticulation as well. It is easy to see that the resulting network is a tree representing $\mathcal{C}\setminus\{x\}$. Moreover, we show that in each case a leaf labeled~$x$ can be added below a new reticulation (possibly with indegree 3) in order to obtain a network~$N'$ that represents~$\mathcal{C}$. Since~$N'$ contains just one reticulation, it is clear that the maximal strict subtrees of~$N'\setminus\{x\}$ are the maximal ST-sets w.r.t. $\mathcal{C}\setminus\{x\}$. {\cass}$(2)$ reconstructs such a network with an indegree-3 reticulation by removing~$x$, removing a dummy taxon~$d$, constructing a tree, adding a leaf labeled~$d$ below a reticulation, adding a leaf labeled~$x$ below a reticulation, removing the leaf labeled~$d$ and contracting the (now redundant) edges between the two reticulations. Note that this works because {\cass}$(2)$ does not collapse in this case.

It remains to consider the possibility that~$X_r$ is a maximal ST-set w.r.t. $\mathcal{C}\setminus\{x\}$. In this case we modify network~$N$ to~$N'$ in such a way that also the other maximal ST-sets w.r.t. $\mathcal{C}\setminus\{x\}$ appear as the leaf-sets of strict subtrees in~$N'\setminus\{x\}$. We again use a case analysis to show that this is always possible in such a way that the resulting network~$N'$ represents~$\mathcal{C}$.
\qed
\end{proof}
\begin{lemma}\label{lem:runningtime}
\textsc{Cass} runs in time~$O(|\mathcal{X}|^{3k+2}\cdot |\mathcal{C}|)$, if~$k$ is fixed.
\end{lemma}
\begin{proof}
Omitted due to space constraints. \qed
\end{proof}
\begin{theorem}
Given a set of clusters~$\mathcal{C}$, \textsc{Cass} constructs in polynomial time a level-2 network representing~$\mathcal{C}$, if such a network exists.
\end{theorem}
\begin{proof}
Follows from Lemmas~\ref{lem:level2} and~\ref{lem:runningtime} and Theorem~\ref{thm:decomp}. \qed
\end{proof}
We conclude this section by showing that for each~$r\geq 2$, there exists a set of clusters~$\mathcal{C}_r$ such that any galled network representing~$\mathcal{C}_r$ needs at least~$r$ reticulations, while \textsc{Cass} constructs a network with just two reticulations, which also represents~$\mathcal{C}_r$. This follows from the following lemma.

\begin{lemma}\label{lem:galledbad}
For each~$r\geq 2$, there exists a set~$\mathcal{C}_r$ of clusters such that there exists a network with two reticulations that represents~$\mathcal{C}_r$ while any galled network representing~$\mathcal{C}_r$ contains at least~$r$ reticulations.
\end{lemma}
\begin{proof}
Omitted due to space constraints. \qed
\end{proof}
\section{Practice}
\begin{table}[t]
\begin{centering}
 \begin{tabularx}{\textwidth}%
     {RR|RRR|RRR}
\hline
\multicolumn{2}{c}{Data} & \multicolumn{3}{c}{\textsc{GalledNetwork}} & \multicolumn{3}{c}{\textsc{Cass}} \\
$|\mathcal{C}|$ & $|\mathcal{X}|$ & $t$ & $k$ & $r$ & $t$ & $k$ & $r$\\
\hline
30  & 5   & 0$s$ & 6  & 6  & 1$s$        & 4 & 4\\
62  & 6   & 0$s$ & 8  & 8  & 7$s$        & 5 & 5\\
126 & 7   & 0$s$ & 10 & 10 & 28$s$       & 6 & 6\\
254 & 8   & 6$s$ & 12 & 12 & 4$m$ 3$s$   & 7 & 7\\
42  & 10  & 0$s$ & 4  & 4  & 6$s$        & 4 & 4\\
38  & 11  & 0$s$ & 7  & 7  & 14$s$       & 5 & 5\\
61  & 11  & 0$s$ & 6  & 6  & 47$s$       & 5 & 5\\
77  & 22  & 0$s$ & 9  & 9  & 36$s$       & 3 & 3\\
75  & 30  & 0$s$ & 11 & 11 & 5$s$        & 2 & 2\\
89  & 31  & 0$s$ & 16 & 16 & 27$m$ 32$s$ & 4 & 4\\
180 & 51  & 0$s$ & 11 & 11 & 30$s$       & 2 & 2\\
193 & 57  & 0$s$ & 1  & 4  & 1$s$        & 1 & 4\\
270 & 76  & 0$s$ & 16 & 16 & 4$m$ 52$s$  & 2 & 2\\
404 & 122 & 1$s$ & 2  & 2  & 21$m$ 10$s$ & 2 & 2\\
\hline
\textbf{135.8} & \textbf{31.9} & \textbf{1}$\bf{s}$ & \textbf{8.5} & \textbf{8.7} & \textbf{4}$\bf{m}$ \textbf{19}$\bf{s}$ & \textbf{3.7} & \textbf{3.9}\\
\hline
\end{tabularx}
\end{centering}\vspace{.2cm}\caption{Results of \textsc{Cass} compared to \textsc{GalledNetwork} for several example cluster sets with~$|\mathcal{C}|$ clusters and~$|\mathcal{X}|$ taxa. For each algorithm, the level~$k$ and reticulation number~$r$ of the output network are given as well as the running time~$t$ in minutes~$m$ and seconds~$s$ on a 1.67$Ghz$ 2$GB$ laptop. The last row gives the average values.}\label{tab:galled}
\end{table}
\begin{table}[t]
\begin{centering}
 \begin{tabularx}{\textwidth}%
     {lR|rR|rR|rRR}
\hline
\multicolumn{2}{c}{Data} & \multicolumn{2}{c}{\textsc{HybridNumber}} & \multicolumn{2}{c}{\textsc{HybridInterleave}} & \multicolumn{3}{c}{\textsc{Cass}}\\
& $|\mathcal{X}|$ & $t$ & $r$ & $t$ & $r$ & $t$ & $k$ &$r$\\
\hline
ndhF and phyB & 40 & 11$h$       & 14 & 23$s$       & 14  & 1$s$       & 4 & 8\\
ndhF and rbcL & 36 & 11$h$ 48$m$ & 13 & 3$s$        & 13  & 0$s$       & 3 & 8\\
ndhF and rpoC & 34 & 26$h$ 18$m$ & 12 & 6$s$        & 12  & 6$s$       & 5 & 9\\
ndhF and waxy & 19 & 5$m$ 20$s$  & 9  & 1$s$        & 9   & 1$s$       & 4 & 6\\
ndhF and ITS  & 46 & $> 2d$      & ?  & 4$m$ 18$s$  & 19  & $> 2d$     & ? & ?\\
phyB and rbcL & 21 & 1$s$        & 4  & 1$s$        & 4   & 0$s$       & 2 & 4\\
phyB and rpoC & 21 & 1$m$ 30$s$  & 7  & 1$s$        & 7   & 0$s$       & 3 & 4\\
phyB and waxy & 14 & 1$s$        & 3  & 1$s$        & 3   & 0$s$       & 2 & 3\\
phyB and ITS  & 30 & 10$s$       & 8  & 1$s$        & 8   & 1$s$       & 4 & 8\\
rbcL and rpoC & 26 & 15$h$ 12$m$ & 13 & 8$s$        & 13  & 10$s$      & 5 & 7\\
rbcL and waxy & 12 & 2$m$ 12$s$  & 7  & 1$s$        & 7   & 1$s$       & 4 & 4\\
rbcL and ITS  & 29 & $> 2d$      & ?  & 10$m$ 12$s$ & 14  & $> 2d$     & ? & ?\\
rpoC and waxy & 10 & 1$s$        & 1  & 1$s$        & 1   & 0$s$       & 1 & 1\\
rpoC and ITS  & 31 & $> 2d$      & ?  & 57$s$       & 15  & $> 2d$     & ? & ?\\
waxy and ITS  & 15 & 10$m$ 20$s$ & 8  & 1$s$        & 8   & 1$s$       & 4 & 5\\
\hline
\textbf{Average} & \textbf{23.2} & \textbf{5}$\bf{h}$ \textbf{22}$\bf{m}$ & \textbf{8.3} & \textbf{4}$\bf{s}$ & \textbf{8.3} & \textbf{2}$\bf{s}$ & \textbf{3.5} & \textbf{5.7}\\
\hline
\end{tabularx}
\end{centering}\vspace{.2cm}\caption{Results of \textsc{Cass} compared to two \textsc{HybridNumber} and \textsc{HybridInterleave} for several combinations of two input trees with~$|\mathcal{X}|$ the number of taxa the two trees have in common, with~$r$ the reticulation number,~$k$ the level and~$t$ the running time in hours~$h$, minutes~$m$ and seconds~$s$. The averages do not include the data sets for which \textsc{HybridNumber} and \textsc{Cass} did not find a solution within two days (denoted ``$>2d$'').}\label{tab:trees}
\end{table}
Our implementation of the \textsc{Cass} algorithm is available as part of the Dendroscope program~\cite{HusonEtAl2007}. To use \textsc{Cass}, first load a set of trees into Dendroscope. Subsequently, run the algorithm by choosing ``options'' and ``network consensus''. The program gives you the option of entering a threshold percentage~$t$. Only clusters that appear in more than~$t$ percent of the input trees will be used as input for \textsc{Cass}. Choose ``minimal network'' to run the \textsc{Cass} algorithm to construct a phylogenetic network representing all clusters that appear in more than~$t$ percent of the input trees.

\textsc{Cass} computes a solution for each biconnected component separately. If the computations for a certain biconnected component take too long, you can choose to ``skip'' the component, in which case the program will quickly compute the cluster network~\cite{HusonRupp2008} for this biconnected component, instead. Alternatively, you can choose to construct a galled network, or to increase the threshold percentage~$t$. For more information on using Dendroscope, see~\cite{HusonEtAl2007}.

We have tested \textsc{Cass} on both practical and artificial data and compared \textsc{Cass} to other programs. The results (using~$t=0$) are in Tables~\ref{tab:galled} and~\ref{tab:trees}. For the former table, several example data sets have been used, which have been selected in such a way as to obtain a good variation in number of taxa, number of clusters and network complexity. For each data set, we have constructed one network using \textsc{Cass}, which we call the \textsc{Cass}-network, and one galled network using the algorithm in~\cite{HusonEtAl2009}. Two conclusions can be drawn from the results. Firstly, \textsc{Cass} uses more time than the galled network algorithm. Nevertheless, the time needed by \textsc{Cass} can still be considered acceptable for phylogenetic analysis. Secondly, \textsc{Cass} constructs a much simpler network in almost all cases. For three data sets, the \textsc{Cass}-network and the galled network have the same reticulation number and the same level. For all other data sets, the \textsc{Cass}-network has a significantly smaller reticulation number, and also a lower level, than the galled network.

\begin{figure}[t]
  \centering
  \includegraphics[width=\textwidth]{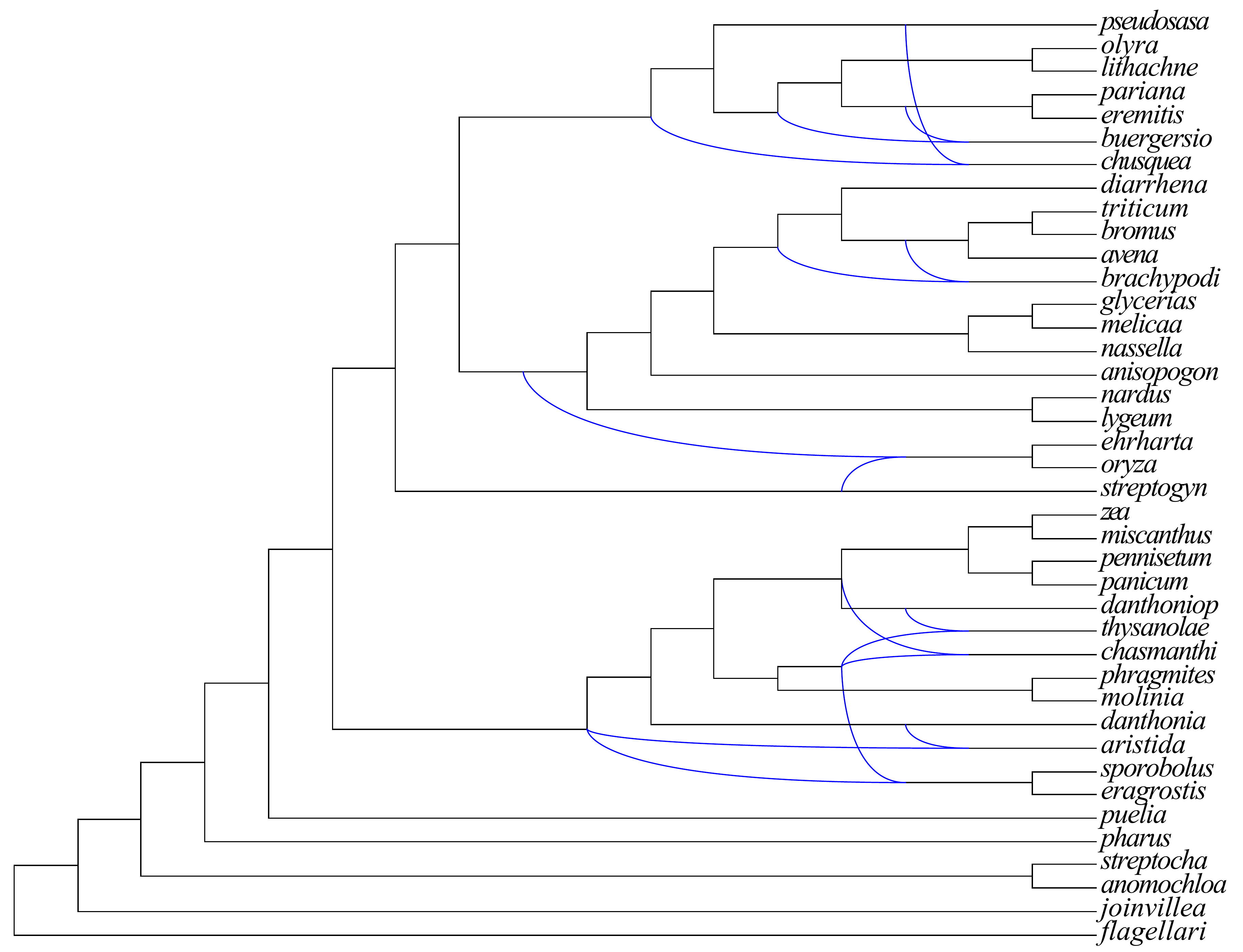}
  \caption{Level-4 network constructed by \textsc{Cass} for the ndhF and phyB trees of the \emph{Poaceae} grass data set, within 1 second.}
  \label{fig:grass}
\end{figure}
In Table~\ref{tab:trees} are the results of an application of \textsc{Cass} to practical data. This data set consists of six phylogenetic trees of grasses of the \emph{Poaceae} family, originally published by the Grass Phylogeny Working Group~\cite{grassgroup}. The phylogenetic trees are based on sequences from six different gene loci, ITS, ndhF, phyB, rbcL, rpoC and waxy, and contain 47, 65, 40, 37, 34 and 19 taxa respectively. We have compared the results of \textsc{Cass} with results of \textsc{HybridNumber}~\cite{BordewichEtAl2007}, which is a program that computes the minimum number of reticulations required to combine two phylogenetic trees (on the same set of taxa) into a phylogenetic network that displays both trees. Since \textsc{HybridNumber} can only be used for pairs of trees with identical taxon sets, we have also applied \textsc{Cass} to pairs of trees and restricted each of the two trees to the taxa that appear in both. Table~\ref{tab:trees} shows that the running time of \textsc{Cass} is much better than the running time of \textsc{HybridNumber}. Very recently, an improved version of \textsc{HybridNumber}, called \textsc{HybridInterleave}~\cite{CollinsEtAl2009}, was published, which is significantly faster than the original program.

Table~\ref{tab:trees} shows that \textsc{Cass} requires significantly fewer reticulations than \textsc{HybridNumber} and \textsc{HybridInterleave}. This is caused by the fact that the latter programs require that a network displays both input trees. The networks constructed by \textsc{Cass} do not necessarily display both input trees, but still represent all clusters from both trees, and use fewer reticulations to do so. Other advantages of \textsc{Cass} are that it can also be used for cluster sets that are obtained from more than two trees and for cluster sets obtained from trees on nonidentical sets of taxa. Moreover, while \textsc{HybridNumber} and \textsc{HybridInterleave} only compute the required number of reticulations, \textsc{Cass} also constructs an actual network. See for example Figure~\ref{fig:grass} for the output network of \textsc{Cass} for the ndhF and phyB trees of the \emph{Poaceae} data set.
\section{Discussion}
We have introduced the \textsc{Cass} algorithm, which can be used to combine any set of clusters into a phylogenetic network representing those clusters. We have shown that the algorithm performs well on practical data. It provides a useful addition to existing software, because it usually constructs a simpler network representing the same set of input clusters. Furthermore, we have shown that \textsc{Cass} provides a polynomial-time algorithm for deciding whether a level-2 phylogenetic network exists that represents a given set of clusters. This algorithm is more useful in practice than algorithms for similar problems that take triplets as input~\cite{TCBB2009,simplicity,JanssonEtAl2006,JanssonSung2006,ToHabib2009}, because the latter algorithms need at least one triplet for each combination of three taxa as input, while \textsc{Cass} can be used for any set of input clusters. Furthermore, \textsc{Cass} is also not restricted to two input trees on identical taxon sets, as the algorithms in~\cite{BordewichEtAl2007,CollinsEtAl2009,LinzSemple2009}. Finally, we remark that \textsc{Cass} can also be used when one or more multi-labeled trees are given as input. One can first compute all clusters in the multi-labeled tree(s) and subsequently use \textsc{Cass} to find a phylogenetic network representing these clusters. Several theoretical problems remain open. First of all, does \textsc{Cass} always construct a minimum-level network, even if this minimum is three or more? Secondly, what is the complexity of constructing a minimum level network, if the minimum level~$k$ is not fixed but part of the input? Is this problem FPT when parameterized by~$k$? Finally, it would be very interesting to design an algorithm that finds a network representing a set of input clusters that has a minimum reticulation number. So far, not even a nontrivial exponential-time algorithm is known for this problem.
\section*{Acknowledgements}
We thank Mike Steel for organizing the Cass workshop in the Cass Field Station in February 2009, where we started this work. Leo van Iersel was funded by the Allan Wilson Centre for Molecular Ecology and Evolution, Steven Kelk by a Computational Life Sciences grant of The Netherlands Organisation for Scientific Research (NWO) and Regula Rupp by the Deutsche Forschungsgemeinschaft (PhyloNet project).
\bibliographystyle{abbrv}

\pagebreak
\appendix
\section{Appendix}

\noindent\textbf{Lemma}~\ref{lem:level2}. Given a set of clusters~$\mathcal{C}$, such that~$IG(\mathcal{C})$ is connected and any~$X\subsetneq\mathcal{X}$ is separated, {\cass}$(1)$ and {\cass}$(2)$ construct a simple level-1 respectively a simple level-2 network representing~$\mathcal{C}$, if such a network exists.
\begin{proof}
We start by proving the following claim. We assume that networks do not contain biconnected components with only one outgoing edge (because such structures are highly redundant). Furthermore, in this proof we identify each leaf with the taxon it is labeled by, to shorten the notation.
\begin{claim}\label{clm:alwayssimple} Given a set of clusters~$\mathcal{C}$, such that~$IG(\mathcal{C})$ is connected and any~$X\subsetneq\mathcal{X}$ is separated, then any network $N$ representing $\mathcal{C}$ is simple and no two leaves in $N$ have the same parent. Additionally, if such a network
$N$ exists, then there also exists a \emph{binary} simple network $N'$ representing {\cset} which has the same level as $N$ and such that no two leaves of $N'$ have the same parent.
\end{claim}
\begin{proof}
If $N$ is not simple then it contains a cut-edge $(v_1, v_2)$ such that $v_2$ is not a leaf and some subset of the taxa $\mathcal{X}' \subsetneq \mathcal{X}$ is reachable by directed paths starting at $v_2$. Since we assume that networks do not contain biconnected components with only one outgoing edge, $|\mathcal{X}'| \geq 2$. Now, given that $\mathcal{X}'$ is below a cut-edge, it follows that for every cluster $C \in \mathcal{C}$ holds that either $\mathcal{X}' \subseteq C$ or $C \subseteq \mathcal{X}'$. So $\mathcal{X}'$ is unseparated, giving us an immediate contradiction. To prove the second half of the lemma we show how to obtain $N'$ from $N$ by expanding edges. First we deal with nodes $v$ that have both indegree and outdegree greater than 1. Here we replace the node $v$ by an edge $(v_1, v_2)$ such that the edges incoming to $v$ now enter $v_1$, and the edges outgoing from $v$ now exit from $v_2$. Subsequently nodes with indegree at most 1, and outdegree $d \geq 3$, can be replaced by a chain of $(d-1)$ nodes of indegree at most 1 and outdegree 2. Nodes with indegree $d \geq 3$ and outdegree 1 can be replaced by a chain of $(d-1)$ nodes of indegree 2 and outdegree 1. These transformations preserve the reticulation number of the network and do not introduce any nontrivial cut-edges (i.e. cut-edges that do not have a leaf as head), so the resulting network $N$ is a binary simple network with the same level as $N'$. Binary simple networks cannot contain sibling leaves, so we are done. \qed
\end{proof}
Before continuing we need some definitions. We say that a tree-edge $e = (v_1, v_2)$ (i.e. an edge where $v_2$ is not a reticulation) of a network $N$ is
\emph{contraction-safe} for $\mathcal{C}$ if $N$ represents $\mathcal{C}$ and there is no $C \in \mathcal{C}$ that is represented by $e$. (An edge $(u,v)$ of $N$ is said to represent a cluster $C$ if there exists a tree $T$ on $\mathcal{X}$ that is displayed by $N$, and such that $C$ consists of all taxa reachable by a directed path from $v$ in~$T$.) Clearly such an edge can be contracted to obtain a new network $N'$ that still represents $\mathcal{C}$.

We are now ready to prove the lemma. Suppose we are given a set of clusters~$\mathcal{C}$, such that~$IG(\mathcal{C})$ is connected and any~$X\subsetneq\mathcal{X}$ is separated.

It is clear that {\cass}$(0)$ will return in polynomial-time a tree $T$ representing $\mathcal{C}$, if it exists. In this case $T$ will
be the unique tree that represents {\cset} and which contains no contraction-safe edges.

We now show that {\cass}$(1)$ will return in polynomial time a simple level-1 network that represents ${\cset}$, if it exists. Suppose
then that such a network, $N$, exists. We assume that $N$ is a binary simple level-1 network. {\cass} will thus at some iteration correctly guess and remove the (unique) leaf $x$ whose parent is a reticulation in $N$. {\cass} will then construct the unique (and in general non-binary) tree $T$ that represents $\mathcal{C} \setminus \{x\}$ and which contains no contraction-safe edges. To complete the level-1 case it is necessary to show that $x$ can be hung back from two edges of $T$ (in the sense of lines 14-16 of the \textsc{Cass} pseudocode) to create a level-1 network representing $\mathcal{C}$. In \cite{TCBB2009} it is described how, after removal of leaves, a binary simple level-$k$ network always has a topology equal to one of the binary \emph{simple level-$k$ generators}, depicted in Figure \ref{fig:generators}. We repeat the following definition~\cite{TCBB2009}.

\begin{definition}
\cite{TCBB2009} A simple level-$k$ network $N$, for $k\geq 1$, is a network obtained by applying the following transformation (``leaf
hanging'') to some simple level-$k$ generator such that the resulting graph is a valid network:
\begin{enumerate}
\item replace each edge $X$ by a path and for each internal node $v$ of the path add a new leaf $x$ and an edge
$(v,x)$; we say that ``leaf $x$ is on side $X$''; and \item for each node $Y$ of indegree 2 and outdegree 0 add a
new leaf $y$ and an edge $(Y,y)$; we say that ``leaf $y$ is on side $Y$''.
\end{enumerate}
\end{definition}

\noindent
Consider in particular that $N$ is constructed from the unique level-1 generator~1. There are two cases. In the first case $N$ has leaves on both sides $A$ and $B$, in which case let $a$ (respectively $b$) be the leaf on side $A$ (respectively $B$) that is furthest from the root. We hang $x$ from the edges in $T$ that feed into $a$ and $b$ respectively, to obtain the network $N'$. We refer to the two corresponding reticulation edges as the $a$- (respectively, $b$-) \emph{reticulation edge}. Consider a cluster $C \in \mathcal{C}$. If $\{x, a\} \subseteq C$ or $\{x,a\} \cap C = \emptyset$ then we see that $N'$ represents $C$ because $T$ represented $C \setminus \{x\}$ and we can simply ``switch on'' the $a$-reticulation edge. If $\{x, a\} \cap C = \{a\}$ then $b\not \in C$, and $b \not \in C \setminus \{x\}$, so we can switch on the $b$-reticulation edge. If $\{x, a\} \cap C = \{x\}$ then $C$ was either a singleton (which is trivially represented) or $b \in C$, in which case we can switch on the $b$-reticulation edge. In the second and final case, assume without loss of generality that only side $A$ contains leaves. Let $a$ be defined as before. We hang $x$ back from the edge
feeding into $a$ in $T$, and from a new root node that we also connect to the old root of $T$. Clusters in $C$ of the form $\{x, a\} \subseteq C$ or $\{x,a\} \cap C = \emptyset$ are dealt with as before. If  $\{x,a\} \cap C = \{x\}$ then $C$ is a singleton. If $\{x,a\} \cap C = \{a\}$ then we can switch the reticulation edge leaving the new root on and we are done.

\begin{figure}[t]
  \centering
  \includegraphics[width=9cm]{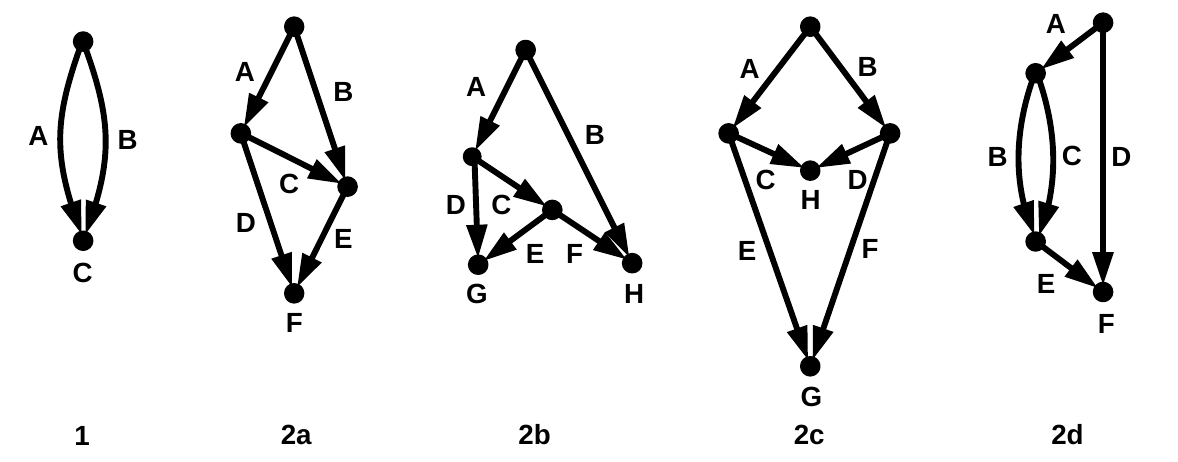}
  \caption{The single level-1 generator and the four level-2 generators.}
  \label{fig:generators}
\end{figure}

There remains only the case that there exists a simple level-2 network representing {\cset}, but no tree or simple level-1 network. This case is rather complex and
requires some extra terminology, although the central idea has much in common with the proof for simple level-1 networks that we have just presented. Given a network~$N$
and reticulation~$r$ with a leaf~$x$ as child, let
$N\setminus \{x\}$ denote the network obtained by
removing the
reticulation~$r$ and the leaf~$x$ from~$N$. We say that a network~$N$ is \emph{drooped with respect to {\cset}} if $N$ represents {\cset} and there exists a leaf~$x$ of~$N$ whose parent is a reticulation, such that the leaf-sets of maximal strict subtrees of $N\setminus\{x\}$ correspond to the maximal ST-sets of $\mathcal{C}\setminus\{x\}$.

We are going to first prove that, if there is a simple level-2 network $N$ representing {\cset}, then there exists a simple level-2 network $N'$ that is drooped w.r.t. {\cset}. We do this partially by case analysis on the four possible generator topologies for $N$ shown in Figure~\ref{fig:generators}. The general strategy will be to argue that either $N$ is already drooped, or that it can be transformed into some new level-2 network $N'$ representing {\cset}. If $N'$ is subsequently level-1 and/or not simple then we obtain a contradiction. Otherwise, $N'$ is (as we will demonstrate) a drooped simple level-2 network.

Consider any leaf~$x$ whose parent is a reticulation of $N$. Observe that $N\setminus\{x\}$ contains exactly one reticulation~$r'$, below which hangs a strict subtree~$T_r$ with leaves~$X_r$ (possibly, $|X_r|=1$ or~$|X_r|=0$). Note that, if $X_r$ is the empty set then $N$ contains an edge between two reticulations and contracting this edge leads to a network~$N'$ that is drooped w.r.t.~$\mathcal{C}$, because $N' \setminus \{x\}$ is a tree. We distinguish two major cases.\\
\\
\emph{First major case: $X_r$ is not a maximal ST-set w.r.t. $\mathcal{C}\setminus\{x\}$ and not the empty set.}\\
\\
In this case it follows that there is some maximal
ST-set~$X$ that contains~$X_r$ and also contains at least one leaf~$\ell$ that is not reachable by a directed path from the reticulation of $N\setminus\{x\}$. We can replace~$\ell$ by a strict subtree on~$X$ that represents $\mathcal{C}|X$. Such a tree exists because~$X$ is an ST-set. We remove all leaves in~$X$ that are not in this strict subtree. Since there are now no leaves left below the reticulation, we can remove this reticulation as well. Let~$N^*$ be the resulting network. It is easy to see that~$N^*$ is a tree representing $\mathcal{C}\setminus\{x\}$. We now require a case-analysis to show how the leaf~$x$ can be hung back into $N^*$ to obtain a network $N'$ that is drooped w.r.t. {\cset}.\\
\\
\noindent\textbf{Case generator 2a}\\
\\
Here we assume that the leaf~$x$ is equal to the leaf on side $F$ of $N$.

Let $p$ be (in $N$) the leaf on side $E$ that is furthest from the root. Such a leaf will definitely exist by assumption that $X_r$ is not empty. The common core of the construction, independent of the exact case, requires locating $p$ in $N^*$ and hanging $x$ back from the edge that feeds into $p$. We call this reticulation edge the $p$-edge. Depending on the exact case we will additionally hang back $x$ from one or two other places (i.e.
to create an indegree-2 or indegree-3 reticulation respectively). But note already that, however we add these additional reticulation edges, a cluster $C \in {\cset}$ such that $\{p,x\} \subseteq C$ or $\{p,x\} \cap C = \emptyset$ will definitely be represented by $N'$. The argument for this is identical to that used in the level-1 case. So we only need to worry about non-singleton clusters $C$ where $\{p,x\} \cap C = \{p\}$ or $\{p, x\} \cap C = \{x\}$. We hang a second reticulation edge of $x$ from the root of $N^*$ and call this the root-edge. Now, let $l$ be the leaf on side $D$ (of $N$) that is furthest from the root of $N$; if such a leaf does not exist then let $l$ be the leaf on side $C$ (of $N$) that
is furthest from the root of $N$, and if that also does not exist then let $l$ be the leaf on side $A$ (of $N$) that is furthest from the root of $N$. For brevity we will henceforth abbreviate this specification of $l$ to ``the lowest leaf on sides $D; C; A$''. If $l$ exists (in general it might not) then hang a third reticulation edge of $x$ from the edge in $N^*$ that feeds into $l$, call this the $l$-edge. Consider then a non-singleton cluster $C$ in {\cset} that contains $x$ but not $p$. Then $l$ exists and cluster $C$ definitely contained it, so $C \setminus \{x\}$ contained $l$ and was represented by $N^*$. So in $N'$ (the network we get after adding~$x$ below a reticulation) we can switch on the $l$-edge to obtain cluster $C$. Finally, consider a cluster $C$ that contained $p$, but not $x$. Then switching on the root-edge is sufficient. So $N'$ represents {\cset}.

If $N'$ is a level-1 network then we have a contradiction. Otherwise it is a drooped
simple level-2 network (because by the earlier claim all networks that represent $C$ are simple).\\
\\
\noindent\textbf{Case generator 2b}\\
\\
Here we assume that the leaf~$x$ is the leaf on side $G$ of $N$. Let $l$ be the lowest leaf on sides $D; A$ of $N$. Let $p$ be the lowest leaf on sides $E; F; C; A$. We hang $x$ back below a new reticulation of indegree 2 or 3.
More specifically: we hang $x$ from the edges that feed into $l, p$ and $h$: the leaf on side~$H$. If $l$ and $p$ both exist and $l \neq p$ then this
reticulation clearly has indegree-3. If $l=p$ or at least one of $l$ and $p$ does not exist then we also hang $p$ back from the root.
(The only situation when the reticulation has indegree-2 is thus if neither $l$ nor $p$ exists). Consider now clusters in {\cset}. We distinguish several cases: in the case that neither $l$ nor $p$ existed we get a level-1 network (and thus a contradiction), otherwise a drooped level-2 network.

Suppose that the leaf $l$ existed. Then clusters $C$ of the form $\{x,l\} \subseteq C$ or $\{x,l\} \cap C = \emptyset$ are (as in earlier cases) easy to deal with. For a non-singleton cluster $C$ such that $\{x,l\} \cap C = \{x\}$ holds that $p \in C$ and/or $h \in C$.
Then in $N'$ we can switch on the $p$-edge or the $h$-edge depending on which one is relevant. For a non-singleton cluster $C$ such that $\{l,x\} \cap C = \{l\}$ holds that $\{h,p\} \cap C = \emptyset$, so in $N'$ we can switch the $h$-edge on.

Suppose, alternatively, that the leaf $p$ existed. Again, ``both $p$ and $x$'' and ``neither $p$ nor $x$'' clusters are easy to deal with. So, consider a cluster that contains $x$ but not $p$. Then this cluster will contain either $l$ or $h$ and we are done. Consider a cluster that contains $p$ but not $x$. If $l$ exists then it will not be in the cluster, so we are done. If $l$ does not exist then we can use the root-edge in $N'$ and we are done.

The final case is that neither $l$ nor $p$ exists. Clusters that contain $h$ and $x$, or neither $h$ nor $x$, are easy to deal with. So, consider a
non-singleton cluster that contains $x$ but not $h$. But the sides $D, A, E, F, C$ are all empty, meaning that the cluster is a singleton, contradiction. For clusters that contain $h$ but do not contain $x$ we can use the root-edge in $N'$, and we are done.\\
\\
\noindent\textbf{Case generator 2c}\\
\\
Here we assume that the leaf~$x$ is the leaf on side $G$. We assume without loss of generality that $N$ contains at least three leaves.
(If $N$ contains only two leaves then $N$ is clearly already drooped).
Let $l$ be the lowest leaf (in $N$) on $E; C; A$. Let $p$ be the lowest leaf (in $N$)
on $F; D; B$.
Note that (by the assumption that there are at least three leaves in $N$) at least one of $l$ and $p$ will exist. If they both exist then hang $x$ back from the edges feeding into $l, p$ and $h$: the leaf on side~$H$. If (without loss of generality) only $l$ exists then hang $x$ back from the edges feeding into $l, h$ and also the root. Consider then non-singleton clusters in {\cset} that contain $x$ but not $l$. Then the cluster either contains $h$ or $p$, and we are done. Finally consider a cluster that contains $l$ but not $x$. If $p$ exists then it will not be in the cluster, so we are done. If $p$ does not exist then we can use the root-edge.\\
\\
\noindent\textbf{Case generator 2d}\\
\\
As in case 2a we assume that there is at least one leaf on side $E$ (because otherwise $N$ was already drooped).  We assume that the leaf $x$ we removed is the leaf on side
$F$ of $N$. Let $l$ be the lowest leaf on side $E$ in $N$ (which must exist). Let $p$ be the lowest leaf on side $D$ in $N$. If $p$ exists then we can hang back $x$ from
the edges feeding into $l$ and $p$. Otherwise we hang $x$ back from the edge feeding into $l$, and the root. Consider then non-singleton clusters in {\cset} that contain
$x$ but not $l$. Then $p$ must exist, and the cluster must contain it. For clusters that contain $l$ but not $x$ we can either use the $p$-reticulation edge (because it is
not possible to contain $l$ and $p$ but not $x$) or use the root reticulation edge in $N'$. $N'$ is thus level-1, and we have a contradiction.

This concludes the case analysis for the generators 2a, 2b, 2c and 2d for the first major case i.e. when~$X_r$ is not a maximal ST-set or the empty
set.\\
\\
\emph{Second major case: $X_r$ \emph{is} a maximal ST-set w.r.t. $\mathcal{C}\setminus\{x\}$.}\\
\\
Here we argue that either $N$ is already drooped w.r.t $\mathcal{C}$ (i.e. for some reticulation leaf $x$ not only $X_r$, but also all other
maximal ST-sets of $\mathcal{C}\setminus \{x\}$, correspond to strict subtrees of $N \setminus \{x\}$)
or that it is possible to transform $N$ into a network $N'$ with this property.
We again use a generator-based case analysis.\\
\\
To shorten the proofs we introduce abbreviations for several commonly-used concepts. If we say \emph{``hang $x$ back from $l_1, \ldots, l_i$''} $(i \geq 2)$
where
$l_1, \ldots, l_i \in \mathcal{X}$ we mean: (1) introduce $x$ into the network as the only child of a new reticulation $r_x$, (2) for each $l_j$ subdivide the
unique edge feeding into $l_j$ to create a new node $v_j$, and finally (3) for each $l_j$ we add the reticulation edge $(v_j, r_x)$. If we
say \emph{``hang $x$ back from $l_1, \ldots, l_i$ and the root''} $(i \geq 1)$ this is defined identically except that after steps (1)-(3) we additionally add
an edge with head $r_x$ and tail the root. As in earlier proofs we will make heavy use of the fact that, if $x$ is hung back from (amongst others) $l_j$ to
obtain a network $N'$, then clusters $C \in \mathcal{C}$ for which $C \cap \{l_j, x\} = \{l_j, x\}$ or $C \cap \{l_j, x\} = \emptyset$, are easily seen
to be represented by $N'$. That is because in such cases the reticulation edge $(v_j, x)$ can simply be ``switched on''. Clusters $C$ of the form
$C \cap \{l_j, x\} = \{l_j\}$ or $C \cap \{l_j, x\} = \{x\}$ require a little more work and in each case will be verifiable by inspection.

If we \emph{``tidy up''} a network we repeatedly apply the following five steps until none is applicable: (1) delete unlabeled nodes with outdegree 0; (2)
suppress
nodes with indegree and outdegree 1 (i.e. contract one edge incident to the node); (3) replace multiple edges by single edges, (4) remove the root if it has
outdegree 1 and (5) contract biconnected components that have only one outgoing edge. Note that tidying up does not affect the set of clusters that a
network represents and is simply a housekeeping measure.

	If we say, \emph{``move maximal ST-sets below cut-edges''} we refer to the following (fundamental) procedure. Suppose $N \setminus \{x\}$ represents
$\mathcal{C} \setminus
\{x\}$. Suppose there is a non-singleton maximal subset $S$ of $\mathcal{C} \setminus \{x\}$ that does not correspond to a strict subtree of $N \setminus\{x\}$. Then we can pick any leaf $l$ of $S$, delete all leaves
$l' \in S \setminus \{ l \}$ from $N \setminus \{x\}$, replace~$l$ with the unique tree that represents precisely those clusters in $\mathcal{C}|S$, and tidy up. This creates a new network in which $S$ \emph{does} appear as a strict subtree and (because $S$
is an ST-set) still represents $\mathcal{C} \setminus \{x \}$. We can repeat this process until all such $S$ appear as strict subtrees of the final network. In other words, until every maximal ST-set is equal to the set of leaves reachable from some cut-edge. Note that, crucially, this procedure will not affect singleton maximal ST-sets or (in this case) $X_r$ because these already correspond to strict subtrees of the network.

	We say \emph{``remove $x$ and transform''} to refer to the combined process of removing $x$ and its parent (from $N$) to obtain $N \setminus \{x\}$,
tidying this network up and subsequently moving all maximal ST-sets (of $\mathcal{C} \setminus \{x\}$) below cut-edges.\\
\\
\noindent\textbf{Case generator 2a}\\
\\
Let $f$ be the leaf on side $F$. Let $a, b, d, e$ be the leaf on respectively side $A, B, D, E$ that is furthest from the root. Let $c$ be the leaf on side $C$ that is
\emph{closest} to the root. Leaf $e$ must exist, because otherwise $N$ was already drooped. We take $x=f$. We distinguish two subcases. In one subcase we construct
a drooped network $N'$ by removing $x$ and transforming, and then hanging $x$ back in such a way that a network is created that represents all
clusters in $\mathcal{C}$. This network will be drooped
(because, prior to hanging $x$ back, we moved the maximal ST-sets of $\mathcal{C}\setminus \{x\}$ under cut-edges) and thus we are done. In the second case we
will show that $N$ was already drooped. For the first case, suppose that leaf $d$ exists. It is easy to see that (after  removing $x$ and transforming) hanging
$x$
back from $e$ and
$d$ creates a drooped network w.r.t. $\mathcal{C}$; the argumentation (e.g. regarding the four possibilities for $C \cap \{e,x\}$ for each $C \in \mathcal{C}$)
is identical to that
used in the previous proofs, and we are done. If $d$
does not exist, and $c$
does not exist, but $a$ does exist, then we hang $x$ back from $e$ and $a$, done. If none of $d,c,a$ exist we hang $x$ back from $e$ and the root, done. This
leaves
us with
the case that $d$ does not exist but $c$ does exist. We observe that hanging $x$ back from $e$ and $c$ creates a drooped network $N'$ that is consistent with all clusters
except with the possible exception of a cluster $C$ that (in $N$) contains $c$ and all leaves on side $E$, but not $f$ or any leaves from side $A$. If such a
cluster
does not exist then we are done.

Assume, then, that it does exist. Observe that if we had hung $x$ back from $e, c$ \emph{and} the root then we would have created a (potentially level-3)
drooped network w.r.t. $\mathcal{C}$. Let $N^{*}$ be the network obtained by tidying up $N \setminus \{x\}$.
We observe that $\mathcal{C}\setminus \{x\}$ contains at most one non-singleton maximal ST-set that is not equal to $X_r$.
Suppose the opposite was true i.e. that
$\mathcal{C}\setminus \{x\}$ contained at least two non-singleton maximal ST-sets not equal to $X_r$. If we then moved maximal ST-sets under cut-edges we would
create a
network with at least two nontrivial cut-edges (excluding the cut-edge associated with the strict subtree corresponding to $X_r$). Hanging $x$ back from $e, c$ and the
root in this network would create a network $N'$ that represents $\mathcal{C}$ but which contains at least one nontrivial cut-edge. But the set of leaves reachable by a
directed path from such a nontrivial cut-edge forms an unseparated subset of $\mathcal{X}$, which by assumption is not possible. Now, if $\mathcal{C} \setminus \{x\}$
contains no non-singleton maximal ST-sets not equal to $X_r$ then we are immediately done, because $N$ was already drooped. So we conclude that there is
exactly one non-singleton maximal ST-set $S$ of $\mathcal{C} \setminus \{x\}$ that is not equal to $X_r$, and that this must contain $c$. Note that, because of
the
existence of cluster $C$ (in particular the fact that cluster $C$ contains leaves from side $E$, and that $X_r$ is already a maximal ST-set), $S$ must be
entirely contained within the leaves of side $C$ (in $N$). $S$ must thus contain at least two leaves $c_1, c_2$ on side $C$. Let $c_1$ and $c_2$ be the leaves
in $S$ that are furthest from the root, with $c_2$ furthest away. Some cluster $C'$ separated $c_1$ from $c_2$ in $\mathcal{C}$ and this proves that $c_1$ and
$c_2$ are the last two leaves on side $C$. (Otherwise $C'\setminus \{x\}$ would prevent $S$ from being a maximal ST-set of $\mathcal{C}\setminus \{x\}$). We
conclude (again, by the separation of $c_1$ and $c_2$) that $C'$ contained leaves from side $E$. But then $C'\setminus\{x\}$ prevents $S$ from being
a maximal ST-set of $\mathcal{C}\setminus\{x\}$. We conclude thus that $\mathcal{C} \setminus \{x\}$ actually contains no non-singleton
maximal ST-sets, with the possible exception of $X_r$. So $N$ was already drooped.\\
\\
\noindent\textbf{Case generator 2b}\\
\\
Let~$g,h$ be the leaves on sides~$G,H$ respectively. Let $a$ be the leaf on side $A$ furthest from the root, and define $b, \ldots , f$ similarly. If $f$ exists
then take $x=h$. We remove $x$ and transform and then hang $x$ back from $f$ and $b$ (if $b$ exists) or otherwise from
$f$ and the root. A simple case-analysis shows that the resulting network is drooped. So assume that side $F$ has no leaves.
If leaf $e$ exists then take $x=g$, remove $x$ and transform as above, and hang $x$ back from $e$ and $d$ if they both exist, otherwise $e$ and the root. This
again gives a network that is drooped w.r.t. $\mathcal{C}$.
So assume side $E$ also contains no leaves. Suppose that side $C$ contains no leaves. Then we take $x=h$ and (after removing $x$ and transforming) hang back
from $b$ and $g$ (if $b$ exists) and otherwise from $g$ and the root. So there is at least one leaf on side~$C$. If $c$ is
the only leaf on side $C$ and none of the clusters $\{c,g\}, \{c,h\}, \{c,g,h\}$ are in $\mathcal{C}$, then we can safely move $c$ to the top of side $D$,
and we are back in the case when side $C$ contains no leaves, done. If side $C$ contains more than one leaf then at least one of $\{c,g\}, \{c,h\}, \{c,g,h\}$
has
to be in $\mathcal{C}$, because otherwise $c$ is unseparated from the leaf immediately above it. If $\{c,g\} \in \mathcal{C}$ then take $x=h$,
otherwise take $x=g$. First suppose $x=h$. Observe that $\{c\}$ is a maximal
ST-set in $\mathcal{C}\setminus\{x\}$ because $\{c,g\} \in \mathcal{C}\setminus \{x\}$ and (by assumption) $\{g\}$ is a maximal ST-set
in $\mathcal{C}\setminus\{x\}$. Thus, when we remove $x$ and transform, then neither $c$ nor $g$ will move in the sense of moving maximal ST-sets under
cut-edges. This is a critical fact. So, after the removal of $x$ and transformation the path of length two in $N$ between the parent $c'$ of $c$ and the parent $g'$ of $g$ will have been suppressed (by the tidying up) to become a single edge $(c', g')$. We will now further expand the current notion of ``hanging $x$ back'', which already defines ``hanging back'' from leaves and the root, to also define hanging back from an edge $(u,v)$. When we hang $x$ back from the edge $(u,v)$ we subdivide $(u,v)$ to create the new node $u'$ and add the reticulation edge $(u', r_{x})$ (where $r_x$ is defined as earlier). We will hang $x$ back from $b$ and $(c',g')$ (if $b$ exists) and otherwise from $(c',g')$ and the root. By inspection it can be verified that this gives a drooped network that represents $\mathcal{C}$. Symmetrically, if $x=g$ then we hang $x$ back from $d$ (if it exists, otherwise the root) and $(c', h')$ where $c'$ is the parent of $c$ and $h'$ is the parent of $h$ (because, again, $\{c\}$ and $\{h\}$ are maximal ST-sets that do not move in the sense of moving maximal ST-sets under cut-edges). Again we obtain a drooped network w.r.t. $\mathcal{C}$, and we are done. Note that the added complexity of this proof comes from the possible presence of clusters $\{g,h\}$ in the input: this is why we have to hang one of the two reticulation edges of $x$ from a carefully identified edge, rather than (as usual) a leaf or the root.\\
\\
\noindent\textbf{Case generator 2c}\\
\\
Assume $a, \ldots, g, h$ are defined as in case 2b. Let $l$ be the lowest leaf (in $N$) on $E; C; A$. Let $p$ be the lowest leaf (in $N$)
on $F; D; B$. At least one of $l$ and $p$ will exist, because we assume that $N$ has at least three leaves. (Otherwise it is already
drooped). Suppose $\{g,h\} \not \in \mathcal{C}$. Then we can take $x=g$, remove $x$ and transform,
and then hang $x$ back from $l$ and $p$ (if they both exist), or from $l$ and the root (if~$l$ exists), or from $p$ and the root (if~$p$ exists).
This will give a drooped network w.r.t. $\mathcal{C}$.
So assume $\{g,h\} \in \mathcal{C}$. We assume then, without loss of generality, that sides $D$ and $F$ contain no leaves. Suppose that side $B$ also contains
no leaves. In this
case we take $x=g$, remove $x$ and transform, then hang $x$ back from $l$ (which must exist) and the edge $(root, h'$) where $h'$ is the parent of $h$.
(Again, this edge was originally a path of length 2 in $N$ that was subsequently suppressed by the tidying-up operation).
This edge is definitely present because $\{h\}$ is (by assumption) a maximal ST-set of
$\mathcal{C} \setminus \{x\}$ and thus remains unaffected by the moving of maximal ST-sets under cut-edges.
This creates a drooped network. So we assume that side $B$ contains at least one leaf i.e. that leaf $b$ exists. If $b$ is the only leaf on side $B$ and none of
the clusters $\{b,g\},
\{b,h\}, \{b,g,h\}$ are in $\mathcal{C}$ then we could move $b$ to the top of side $A$ and we are back in the case that side $B$ contains no leaves, done. If
side $B$ contains
more than one leaf then at least one of those clusters must be present, otherwise $b$ and the leaf immediately above it were unseparated.
If $\{b,h\} \in \mathcal{C}$ take $x=g$,
otherwise take $x=h$. First suppose $x=g$. As in case 2b we argue that $\{b\}$ and $\{h\}$ are maximal ST-sets in $\mathcal{C} \setminus \{x\}$ and thus that
the edge (i.e. suppressed path) connecting the parent $b'$ of $b$ to the parent $h'$ of $h$ is unaffected
by the movement of maximal ST-sets. In this case we hang $x$ back from $(b',h')$
and from $l$ (if it exists: otherwise the root). It is not too difficult to verify that the resulting network is drooped w.r.t.
$\mathcal{C}$.
The case $x=h$ is almost entirely symmetrical except that we must redefine $l$ to be the lowest leaf on $C; E; A$ (instead of $E; C; A$).\\
\\
\noindent\textbf{Case generator 2d}\\
\\
Take $x=f$, where $f$ is the leaf on side $F$. Let $e$ be the leaf on side $E$ that is furthest from the root. (Leaf $e$ must exist because otherwise $X_r=\emptyset$). Let $d$ be the leaf on side $D$ that is furthest from the root. We remove $x$ and transform, and hang $x$ back from
$e$ and $d$ (if $d$ exists) and otherwise from $e$ and the root. This creates a drooped network.

This concludes the second major case, and we have thus proven that a drooped simple level-2 network $N'$ exists that represents {\cset}. To complete the overall proof we need to show that {\cass} will (re)construct $N'$, or some other simple level-2 network representing {\cset}. Suppose we contract all contraction-safe edges of $N'$ to obtain $N''$. In some iteration, {\cass} correctly identifies a leaf~$x$ whose parent is a reticulation~$r$ and the maximal ST-set~$X_r$ which is the set of leaves below the only reticulation~$r'$ of~$N''\setminus\{x\}$.

Let~$T''$ be the tree obtained by removing~$x,r,X_r$ and~$r'$ from $N''$ and contracting edges entering unlabeled tree-nodes with outdegree at most 1. We first show that~$T''$ is identical to the unique tree $T$ that represents ${\cset} \setminus (X_r \cup \{x\})$ and which contains no contraction-safe edges. ($T$ is the tree that {\cass} constructs in its innermost iteration). Suppose $T'' \neq T$. Then $T''$ will contain at least one contraction-safe edge, implying that $N''$ also contains at least one contraction-safe edge, yielding a contradiction.

If~$X_r\neq\emptyset$, one can reconstruct~$N''$ from~$T''$ as follows. First add a tree representing~$\mathcal{C}|X_r$ below a reticulation and subsequently add~$x$ below another reticulation. Notice however that {\cass} always adds the reticulations below nodes inserted into edges, while in~$N''$ a reticulation might be a child of a node~$v$ with indegree one and outdegree at least two (observe that~$v$ cannot be a reticulation because~$X_r\neq\emptyset$ and that~$v$ cannot have indegree 0 because {\cass} adds a dummy root with an edge to the old root). {\cass} adds the new reticulation below a node inserted into the edge entering~$v$ instead of below~$v$, which leads to a network~$N'''$ that also represents~$\mathcal{C}$.

To conclude the proof, consider the case~$X_r=\emptyset$. Assume without loss of generality that~$N''$ contains one reticulation with indegree 3 (if it contains two reticulations with indegree 2 then we can contract the edge between these reticulations). In this case, {\cass} constructs a network~$N'''$ representing~$\mathcal{C}$ from~$T''$ as follows. It first adds a dummy leaf~$d$ below a reticulation, then it adds~$x$ below another reticulation. This second reticulation is added below nodes inserted into the edge entering~$d$ and one other edge. Before outputting the network, {\cass} removes~$d$ and contracts the edges between the two reticulations. As in the previous case, whenever in~$N''$ the reticulation hangs below a node with outdegree at least two, {\cass} hangs the reticulation below a node inserted into the edge entering~$v$ instead. The resulting network~$N'''$ represents~$\mathcal{C}$. This concludes the proof.
\qed \end{proof}

\noindent\textbf{Lemma}~\ref{lem:runningtime}.
\textsc{Cass} runs in time~$O(|\mathcal{X}|^{3k+2}\cdot |\mathcal{C}|)$, if~$k$ is fixed.
\begin{proof}
Let~$n=|\mathcal{X}|$ and~$m=|\mathcal{C}|$. We analyze the running time of constructing a simple level-$\leq k$ network, since all other computations can clearly be done in~$O(n^{3k+2}\cdot m)$ time. We will show by induction on~$k'$ that a call to \textsc{Cass}$(\mathcal{C},\mathcal{X},k,k')$ takes at most~$O(n^{3k'+2}\cdot m)$ time and returns at most~$O(n^{3k'})$ networks, for fixed~$k'$. The lemma will follow because in the original call~$k'=k$.

For~$k'=0$, \textsc{Cass} only checks if there exists a tree representing~$\mathcal{C}$, which can clearly be done in~$O(n^2\cdot m)$ time and leads to at most one network. Suppose~$k'\geq 1$. The algorithm loops through~$O(n)$ taxa and at most~$O(n^{3k'-3})$ recursively created networks. For each network, the algorithm loops through all pairs of edges. For fixed~$k'$, each network contains~$O(n)$ edges, since a tree contains at most~$2n-2$ edges and the algorithm adds a constant number of edges in each iteration. For each combination of edges, \textsc{Cass} checks if the constructed network represents all clusters. This can be done by looping through the at most~$2^{k'}$ ways of switching edges on and off and checking if all~$m$ clusters are represented by one of the resulting trees, in~$O(2^{k'}\cdot m\cdot n^2)$ time. This is the bottleneck of all computations. Thus, the total time needed by \textsc{Cass} is~O$(n\cdot n^{3k'-3}\cdot n^2\cdot 2^{k'}\cdot m\cdot n^2)$, which is~$O(n^{3k'+2}\cdot m)$ for fixed~$k'$. Similarly, the number of constructed networks is at most~$O(n\cdot n^{3k'-3}\cdot n^2)$ and hence~$O(n^{3k'})$. \qed
\end{proof}

\noindent\textbf{Lemma}~\ref{lem:galledbad}.
For each~$r\geq 2$, there exists a set~$\mathcal{C}_r$ of clusters such that there exists a network with two reticulations that represents~$\mathcal{C}_r$ while any galled network representing~$\mathcal{C}_r$ contains at least~$r$ reticulations.
\begin{proof}
Consider a simple level-2 network~$N$ of type $2a$, with~$r$ leaves on each of the sides~$B$,~$C$ and~$E$ and a single leaf on side~$F$. Let~$\mathcal{C}_r$ be the set of all clusters represented by~$N$. Suppose that there exists a galled network~$N'$ representing~$\mathcal{C}_r$ and containing~$r'<r$ reticulations. It is easy to check that the incompatibility graph of~$\mathcal{C}_r$ (excluding singleton clusters) is connected and hence that~$N'$ contains just one biconnected component (except for the leaves). Thus, the~$r'$ reticulations of~$N'$ each have a leaf as child. Let~$\mathcal{C}'$ be the result of removing the~$r'$ taxa labeling these leaves from~$\mathcal{C}_r$. It follows that there exists a tree representing~$\mathcal{C}'$ and hence that~$\mathcal{C}'$ is compatible. However,~$\mathcal{C}'$ clearly contains at least one leaf on each of the sides~$B$,~$C$ and~$E$ of~$N$, say a leaf~$b$ on side~$B$,~$c$ on side~$C$ and~$e$ on side~$E$. Hence, there will be a cluster~$X\in\mathcal{C}'$ containing~$b$ and~$e$ but not~$c$ and a cluster~$Y\in\mathcal{C}'$ containing~$c$ and~$e$ but not~$b$. It follows that~$X$ and~$Y$ are incompatible, which is a contradiction because we have already shown that~$\mathcal{C}'$ is compatible. \qed
\end{proof}

\end{document}